\documentclass[12pt]{article}
\usepackage{amscd,amsmath,amssymb,amstext,amsthm,exscale,latexsym}
\newcommand {\al}   {\alpha}       \newcommand {\bt}  {\beta}
\newcommand {\g }   {\gamma}       
\newcommand {\dl}   {\delta}       \newcommand {\e }  {\epsilon}

\newcommand {\Lm}   {\Lambda}      
       
\newcommand {\pl}   {\partial}     \newcommand {\nb}  {\nabla}

  \newcommand {\Sl}  {{\textsc{l}}}
  \newcommand {\St}  {{\textsc{t}}}

\renewcommand {\det}{{\sf\,det\,}}

\newcommand   {\const}{{\sf\,const}}     \newcommand   {\diag}{{\sf\,diag\,}}
   
\newcommand {\MO}  {{\mathbb O}}   \newcommand {\MR}  {{\mathbb R}}
\newcommand {\MS}  {{\mathbb S}}

\newtheorem{theorem}{Theorem}[section]

\begin{document}
\title    {Lorentz Invariant Vacuum Solutions in General Relativity}
\author    {M. O. Katanaev
            \thanks{E-mail: katanaev@mi.ras.ru. This work is supported by the
            RSF under a grant 14-50-00005.}\\ \\
            \sl Steklov Mathematical Institute,\\
            \sl ul.~Gubkina, 8, Moscow, 119991}
\date      {07 March 2015}
\maketitle
\begin{abstract}
All Lorentz invariant solutions of vacuum Einstein's equations are found. It is
proved that these solutions describe space-times of constant curvature.
\end{abstract}

Derivation of exact solutions of Einstein's equations is one of the main
problems in general relativity. Many classes of exact solutions are given in
\cite{KrStMaHe80}. As a rule, solutions are derived under the assumption of
some symmetry because Einstein's equations are very complicated. This simplifies
the system of equations and opens a possibility to look for solutions.

De Sitter and anti-de Sitter solutions \cite{deSitt16,deSitt17} were among the
first exact cosmological solutions of vacuum equations with a cosmological
constant. The De Sitter solution describes the space-time of constant
curvature and is invariant with respect to the action of the Lorentz group
$\MS\MO(1,4)$. The Ant-de Sitter solution corresponds to the constant curvature
space and is invariant with respect to the group $\MS\MO(2,3)$. Both symmetry
groups contain the Lorentz subgroup $\MS\MO(1,3)$ of lower dimension. In the
present paper, we prove that all vacuum solutions of Einstein's equations
invariant with respect to the action of the Lorentz group $\MS\MO(1,3)$ describe
space-times of constant curvature and therefor reduce either to de Sitter or to
anti-de Sitter solutions depending on the sign of the cosmological constant.

Consider the Minkowski space-time $\MR^{1,n-1}$ of arbitrary dimension $n$.
The Lorentz metric $\eta_{\al\bt}:=\diag(+-\dotsc-)$ is given in the Cartesian
coordinate system $x^\al$, $\al=0,1,\dotsc,n-1$. This metric is invariant with
respect to the Poincar\'e group and, in particular, with respect to the Lorentz
group. Corresponding Killing vector fields are
\begin{equation}                                                  \label{ekilor}
  K_{\e\dl}=\frac12\left(x_\dl\pl_\e-x_\e\pl_\dl\right)
  =\frac12\left(x_\dl\dl_\e^\g-x_\e\dl_\dl^\g\right)\pl_\g,
\end{equation}
where the indices $\e$ and $\dl$ number $n(n-1)/2$ Killing vectors and
$x_\al:=x^\bt\eta_{\bt\al}$. Let another metric $g_{\al\bt}(x)$ be given in the
Minkowski space-time $\MR^{1,n-1}$. We pose the following problem: find all
metrics $g_{\al\bt}$ invariant under the action of the Lorentz group
$\MS\MO(1,n-1)$.

The equations
\begin{equation*}
  \nb_\al K_\bt+\nb_\bt K_\al=0
\end{equation*}
for the Killing vector fields (\ref{ekilor}) take the form
\begin{equation}                                                  \label{ekillo}
  g_{\al\e}\eta_{\bt\dl}-g_{\al\dl}\eta_{\bt\e}+g_{\bt\e}\eta_{\al\dl}
  -g_{\bt\dl}\eta_{\al\e}+x_\dl\pl_\e g_{\al\bt}-x_\e\pl_\dl g_{\al\bt}=0.
\end{equation}

The metric components $g_{\al\bt}$ are components of a covariant second rank
tensor under the Lorentz transformations. They must be constructed from the
Lorentz metric $\eta_{\al\bt}$ and point coordinates $x=\lbrace x^\al\rbrace$.
The only possibility is the metric of the form
\begin{equation*}
  g_{\al\bt}=A\eta_{\al\bt}+Bx_\al x_\bt,
\end{equation*}
where $A$ and $B$ are some functions on $\MR^{1,n-1}$. The substitution of this
metric into the Killing equations (\ref{ekillo}) restricts the form of the
functions $A$ and $B$. One may prove that they can be arbitrary functions of
only one variable
$$
  s:=x^\al x^\bt\eta_{\al\bt},
$$
which is invariant under the Lorentz transformations. Thus, the metric invariant
with respect to the Lorentz transformations is parameterized by two arbitrary
functions $A(s)$ and $B(s)$. It is useful to rewrite it in another form
\begin{equation}                                                  \label{elocom}
  g_{\al\bt}=f(s)\Pi^\St_{\al\bt}+g(s)\Pi^\Sl_{\al\bt}=
             f\eta_{\al\bt}+(g-f)\frac{x_\al x_\bt}s,
\end{equation}
where $\Pi^\St$ and $\Pi^\Sl$ are projection operators:
\begin{equation*}
  \Pi^\St_{\al\bt}:=\eta_{\al\bt}-\frac{x_\al x_\bt}s,\qquad
  \Pi^\Sl_{\al\bt}:=\frac{x_\al x_\bt}s.
\end{equation*}

The determinant of the metric (\ref{elocom}) is easily computed:
\begin{equation}                                                  \label{edelme}
  \det g_{\al\bt}=(-f)^{n-1}g.
\end{equation}
Therefore, the Lorentz invariant metric is degenerate if and only if $fg=0$. We
suppose that functions $f$ and $g$ are sufficiently smooth, and $f>0$ and
$g\ne0$. Moreover we suppose that there exists a limit
\begin{equation*}
  \underset{s\to 0}\lim\frac{f(s)-g(s)}s,
\end{equation*}
which is necessary for the metric to be defined for $s=0$.

We use the Lorentz metric $\eta_{\al\bt}$ to raise and lower indices unless
otherwise stated.

The Lorentz invariant metric (\ref{elocom}) for $f=g$ was considered by
V.~A.~Fock \cite{Fock61R}.

The invariant interval
$$
  ds^2=fdx_\al dx^\al+(g-f)\frac{(x_\al dx^\al)^2}s
$$
corresponds to the metric (\ref{elocom}).

The metric tensor (\ref{elocom}) has the same form in all coordinate systems
related by Lorentz transformations. However, it changes under translations
$x^\al\mapsto x^\al+a^\al$ because the metric depends explicitly on coordinates
and the origin of the coordinate system is distinguished.

The expression for the metric (\ref{elocom}) in terms of projection operators is
useful because the inverse metric has a simple form
\begin{equation}                                                  \label{einmef}
  g^{\al\bt}=\frac1f\Pi^{\St\al\bt}+\frac1g\Pi^{\Sl\al\bt}.
\end{equation}

For $g>0$, the vierbein
\begin{equation}                                                  \label{eveilo}
  e_\al{}^a=\sqrt f\dl_\al^a+(\sqrt{g\vphantom{f}}-\sqrt f)\frac{x_\al x^a}s.
\end{equation}
can be attributed to the metric (\ref{elocom}).

One can easily check the following properties of the projection operators:
\begin{align*}
  \Pi^{\St\al\bt}x_\bt&=0,    &\Pi^\St_\al{}^\al&=n-1,
  &\pl_\al\Pi^\St_{\bt\g}&=-\frac{\Pi^\St_{\al\bt}x_\g+\Pi^\St_{\al\g}x_\bt}s,
\\
  \Pi^{\Sl\al\bt}x_\bt&=x^\al,&\Pi^\Sl_\al{}^\al&=1,
  &\pl_\al\Pi^\Sl_{\bt\g}&=~~\frac{\Pi^\St_{\al\bt}x_\g+\Pi^\St_{\al\g}x_\bt}s,
\end{align*}
which will be used in calculations below.

Simple calculations yield Christoffel's symbols for the metric (\ref{elocom})
\begin{equation}                                                  \label{echcoa}
  \Gamma_{\al\bt}{}^\g=\frac{f'}f(x_\al\Pi^\St_\bt{}^\g+x_\bt\Pi^\St_\al{}^\g)
  +\frac{g'}g(x_\al\Pi^\Sl_\bt{}^\g+x_\bt\Pi^\Sl_\al{}^\g
  -x^\g\Pi^\Sl_{\al\bt})+\frac{g-f-f's}{sg} x^\g\Pi^\St_{\al\bt},
\end{equation}
where the prime denotes differentiation with respect to $s$. The curvature
tensor for the metric (\ref{elocom}) is
\begin{align}                                                          \nonumber
  R_{\al\bt\g}{}^\dl=&\Pi^\St_{\al\g}\Pi^\St_\bt{}^\dl
  \left[\frac{(f+f's)^2}{sfg}-\frac1s\right]+
\\                                                                     \nonumber
  &+\Pi^\Sl_{\al\g}\Pi^\St_\bt{}^\dl
  \left[2\left(\frac{f+f's}f\right)^\prime+\left(\frac{f'}f-\frac{g'}g\right)
  \frac{f+f's}f\right]+
\\                                                                \label{eculoc}
  &+\Pi^\Sl_\al{}^\dl\Pi^\St_{\bt\g}
  \left[-2\left(\frac{f+f's}g\right)^\prime+\left(\frac{f'}f-\frac{g'}g\right)
  \frac{f+f's}g\right]-(\al\leftrightarrow\bt),
\end{align}
where $(\al\leftrightarrow\bt)$ stands for the preceding terms with the
exchanged indices. Contracting this expression in the indices $\bt$ and $\dl$
yields the Ricci tensor
\begin{align}                                                          \nonumber
  R_{\al\bt}=&\Pi^\St_{\al\bt}
  \left[\frac{n-2}s\left(\frac{(f+f's)^2}{fg}-1\right)
  +2\frac{(f+f's)'}g-\left(\frac{f'}f+\frac{g'}g\right)
  \frac{f+f's}g\right]+
\\                                                                \label{ericom}
  &+\Pi^\Sl_{\al\bt}(n-1)\left[2\frac{(f+f's)'}f
  -\left(\frac{f'}f+\frac{g'}g\right)\frac{f+f's}f\right].
\end{align}
The further contraction with the inverse metric (\ref{einmef}) gives the scalar
curvature
\begin{equation}                                                  \label{escurf}
  R=(n-1)\left[\frac{n-2}{fs}\left(\frac{(f+f's)^2}{fg}-1\right)
  +4\frac{(f+f's)'}{fg}-2\left(\frac{f'}f+\frac{g'}g\right)
  \frac{f+f's}{fg}\right].
\end{equation}

Constant curvature spaces defined by the equation
\begin{equation}                                                  \label{ecolor}
  R_{\al\bt\g\dl}=-\frac{2K}{n(n-1)}(g_{\al\g}g_{\bt\dl}-g_{\al\dl}g_{\bt\g}),
\end{equation}
with some constant $K$ automatically satisfy the vacuum Einstein's equations
with cosmological constant. Let us solve equation (\ref{ecolor}) for the Lorentz
invariant metric (\ref{elocom}). To this end, we lower the last index of the
curvature tensor (\ref{eculoc}) using the metric (\ref{elocom})
\begin{align}                                                     \label{eculoi}
  R_{\al\bt\g\dl}=&\Pi^\St_{\al\g}\Pi^\St_{\bt\dl}\frac1s
  \left[\frac{(f+f's)^2}g-f\right]+
\\                                                                     \nonumber
  &+(\Pi^\Sl_{\al\g}\Pi^\St_{\bt\dl}-\Pi^\Sl_{\al\dl}\Pi^\St_{\bt\g})
  \left[2(f+f's)'-\left(\frac{f'}f+\frac{g'}g\right)(f+f's)\right]
  -(\al\leftrightarrow\bt)
\end{align}
and substitute it into equation (\ref{ecolor}). As a result, we get the system
of differential equations for the functions $f$ and $g$:
\begin{align}                                                     \label{efgeqo}
  \frac{(f+f's)^2}{sg}-\frac fs&=-\frac{2K}{n(n-1)}f^2,
\\                                                                \label{efgeqt}
  2(f+f's)'-\left(\frac{f'}f+\frac{g'}g\right)(f+f's)&=-\frac{2K}{n(n-1)}fg.
\end{align}
The first equation yields a solution for the function $g$:
\begin{equation}                                                  \label{egfexe}
  g=\frac{(f+f's)^2}{f\left(1-\frac{2K}{n(n-1)}fs\right)}.
\end{equation}
Since the inequality $g\ne0$ must be fulfilled for the metric to be
nondegenerate, the function $f$ should satisfy the inequality
\begin{equation}                                                  \label{enfufu}
  f\ne\frac{n(n-1)}{2Ks},\qquad s\ne0.
\end{equation}
The substitution of expression (\ref{egfexe}) into the second equation
(\ref{efgeqt}) yields the identity. Thus, we have proved the first part of the
following statement:
\begin{theorem}
The Lorentz invariant metric
\begin{equation}                                                  \label{eloinm}
  g_{\al\bt}=f\Pi^\St_{\al\bt}
  +\frac{(f+f's)^2}{f\left(1-\frac{2K}{n(n-1)}fs\right)}\Pi^\Sl_{\al\bt},
\end{equation}
where $f(s)$ is an arbitrary positive function satisfying equation
(\ref{enfufu}), is the metric of a constant curvature space. Conversely, the
metric of a constant curvature space can be written in the Lorentz invariant
form (\ref{eloinm}) for some function $f(s)$.
\end{theorem}
\begin{proof}
It remains to prove that any metric of a constant curvature space can be
transformed into the Lorentz invariant form (\ref{elocom}). To show this,
we write the metric (\ref{eloinm}) in a more familiar form. To this end we fix
the function $f$ by putting $f=g$ in the initial representation (\ref{elocom}).
Then equation (\ref{egfexe}) yields the differential equation
$$
  f^{\prime2}s+2f'f+\frac{2K}{n(n-1)}f^3=0,
$$
which has a general solution
$$
  f=\frac C{(C+\frac K{2n(n-1)} s)^2},\qquad C=\const.
$$
The integration constant can be removed by rescaling coordinates. Therefore, we
put $C=1$ without loss of generality. As a result we obtain the metric of
constant curvature
\begin{equation}                                                  \label{ecomeu}
  g_{\al\bt}=\frac{\eta_{\al\bt}}{(1+\frac K{2n(n-1)}s)^2}.
\end{equation}
The fact that any constant curvature metric can be reduced to this form is well
known. The proof of this fact is nontrivial (see, e.g.\ \cite{Wolf72}).
\end{proof}
The performed calculations can be easily transferred to the Euclidean space
metric which is invariant with respect to $\MS\MO(n)$ rotations. In order to do
this, one should replace the Lorentzian metric $\eta_{\al\bt}$ by the Euclidean
metric $\dl_{\al\bt}$ in all formulas.

Since
\begin{equation*}
  \underset{s\to0}\lim\frac{g-f}s=2f'+\frac{2K}{n(n-1)}f^2,
\end{equation*}
the expression for the metric (\ref{eloinm}) is defined for $s=0$ as well.

Now we solve the vacuum Einstein equations with the cosmological constant
$$
  R_{\al\bt}=\Lm g_{\al\bt}
$$
for a Lorentz invariant metric. One could expect that these equations have
solutions not only of constant curvature because their number is less then
the number of equations in the constant curvature condition (\ref{ecolor}).
However, for Lorentz invariant metrics the classes of solutions coincide.
Indeed, the substitution of the Ricci tensor (\ref{ericom}) into Einstein's
equations yields the following system of equations:
\begin{align}                                                     \label{einone}
  \frac{n-2}s\left[\frac{(f+f's)^2}{fg}-1\right]
  +2\frac{(f+f's)'}g-\left(\frac{f'}f+\frac{g'}g\right)\frac{f+f's}g
  &=\Lm f,
\\                                                                \label{eintwo}
  (n-1)\left[2\frac{(f+f's)'}f-\left(\frac{f'}f+\frac{g'}g\right)
  \frac{f+f's}f\right] &=\Lm g.
\end{align}
The second equation coincides with equation (\ref{efgeqt}) for
$$
  \Lm=-\frac{2K}n.
$$
The linear combination of equations (\ref{einone}) and (\ref{eintwo}) with the
coefficients $1/f$ and $-1/g$, respectively, is equivalent to equation
(\ref{efgeqo}).

Thus, we proved that all Lorentz invariant solutions of the vacuum Einstein
equations with cosmological constant are exhausted by constant curvature spaces.

This work is supported by the Russian Science Foundation under grant
14--50--00005.

\end{document}